\documentclass[journal, twocolumn]{IEEEtran}
\usepackage{cite}

\ifCLASSINFOpdf
   \usepackage[pdftex]{graphicx}
  \DeclareGraphicsExtensions{.pdf,.jpeg,.png}
\else
  \usepackage[dvips]{graphicx}
  \DeclareGraphicsExtensions{.eps}
\fi

\usepackage[cmex10]{amsmath}
\usepackage{amsthm}

\usepackage{algorithmic}

\usepackage{array}

\usepackage{fixltx2e}

\newtheorem{prop}{Proposition}
\usepackage{dblfloatfix}

\ifCLASSOPTIONcaptionsoff
  \usepackage[nomarkers]{endfloat}
 \let\MYoriglatexcaption\caption
 \renewcommand{\caption}[2][\relax]{\MYoriglatexcaption[#2]{#2}}
\fi

\usepackage{url}
\usepackage{savesym}
\savesymbol{appendices}
\usepackage{epsfig}
\usepackage{epstopdf}
\usepackage{multicol}
\usepackage{paralist} 
\usepackage{verbatim}
\usepackage{amssymb}

\usepackage{mdwmath}
\usepackage{mdwtab}
\usepackage[section]{placeins}
\usepackage{caption}
\usepackage{subcaption}
\usepackage{xcolor}
\usepackage{color,soul}
\usepackage{appendix, apptools}
\restoresymbol{appendix}{appendices}
\begin{document}

\title{On the Performance of Quickest Detection Spectrum Sensing: The Case of Cumulative Sum}
\author{ 
Ahmed~Badawy, \IEEEmembership{Member,~IEEE,} Ahmed El Shafie, \IEEEmembership{Senior Member,~IEEE} and Tamer~Khattab, \IEEEmembership{Senior Member,~IEEE}
\thanks{A. Badawy was with Politecnico di Torino, DET. (e-mail: ahmed.badawy@polito.it). Ahmed El Shafie is with Qualcomm Tech. Inc, San Diego, CA 92121 USA (e-mail: ahmed.salahelshafie@gmail.com). Tamer Khattab is with Qatar University, Electrical Engineering Department. (e-mail: tkhattab@ieee.org).}

\thanks{This research work is funded by Qatar National Research Fund, QNRF (a member of Qatar Foundation, QF) under grant number NPRP 7-923-2-344. The statements made herein are the sole responsibility of the authors.}
}
\maketitle

\begin{abstract}
Quickest change detection (QCD) is a fundamental problem in many applications. Given a sequence of measurements that exhibits two different distributions around a certain flipping point, the goal is to detect the change in distribution around the flipping point as quickly as possible. The QCD problem appears in many practical applications, e.g., quality control, power system line outage detection, spectrum reuse, and resource allocation and scheduling. In this paper, we focus on spectrum sensing as our application since it is a critical process for proper functionality of cognitive radio networks. Relying on the cumulative sum (CUSUM), we derive the probability of detection and the probability of false alarm of CUSUM based spectrum sensing. We show the correctness of our derivations using numerical simulations.
\end{abstract}


\IEEEpeerreviewmaketitle

\section{Introduction}
The increasing demand of spectrum slots is a result of the exponential growth of wireless networks. On the other hand, this growth is facing the classical spectrum scarcity problem.  Statistical analysis of spectrum usage presented in \cite{Haykin} shows that the spectrum is underutilized. Therefore, interest in cognitive radio (and multi-tier priority access \cite{kazi2019next}) networks has also grown accordingly.  In cognitive radio networks, spectrum slots are allocated to users in a dynamic fashion. At first, the spectrum slot is assigned to its owner, also known as primary user (PU). Users with less priority, also known as secondary users (SU) are allowed to access this designated spectrum slot whenever its owner is not exploiting it.
s
Spectrum sensing is a cornerstone in the deployment of cognitive radio networks. Spectrum sensing can be achieved through different techniques including energy detection \cite{8070358} and cyclostationary detection \cite{8746083, 6951976}. On the other hand, signal detection based on probabilistic models, i.e., maximum-likelihood-ratio test (MLRT) and general-likelihood-ratio test (GLRT) \cite{Kay,poorbook, Poor1, PoorJournal, 7565143} exploits the distributions of the received signal under the two hypotheses (occupied or vacant spectrum slot) to decide on the presence or absence of the PU's signal. Moreover, spectrum sensing can be applied in local or cooperative fashions \cite{7247243}.


One critical problem in detection theory is the quickest change detection (QCD) problem. The objective of QCD is to detect the change point in a series of collected samples or measurements as quickly as possible, i.e., finding the point at which the distribution of the received samples changes. Applications of QCD are numerous, which includes spectrum sensing \cite{Poor1}, resource allocation and scheduling \cite{ren2017quickest}, power system line outage detection \cite{7041234} and  bioinformatics \cite{muggeo2010efficient}. 


A framework for sequential detection for cognitive radio networks is presented in \cite{Poor1}. In \cite{7355347} the authors derived an approximate closed-form expression for the distribution of the detection delay for quickest detection. A joint design based on  observation scheduling policy and stopping time that minimizes the detection delay for quickest detection is presented in \cite{ren2017quickest}. Furthermore, the authors in \cite{ren2017quickest} extended their study to the multi-channel sensing case.

	
In this paper, motivated by the great need to find closed-form expressions for false-alarm and detection probabilities for the above-mentioned critical applications, we revisit the problem in \cite{Poor1} and provide closed-form expressions for the false-alarm and detection probabilities under finite sensing interval. To the best of the authors knowledge, closed-form expressions for these probabilities in the considered problem do not exist in literature. In fact, it was stated in \cite{7355347} that exact analysis for this problem is intractable. The sought expressions are important in practical systems since the number of collected samples is finite and any quality-of-service optimization will require the knowledge of both detection and false-alarm probabilities. We give those probabilities in closed-form and verify all our findings through numerical evaluations.
		
\section{System Model} \label{sec_sys_model}
We consider an SU operating in frame basis \cite{LULU12}, as depicted in Fig. \ref{sys_model}. The time is partitioned into frames of equal length. Each frame consists of a spectrum sensing phase and data transmission phase. In case the decision during the spectrum sensing phase is declared to be existence of the PU's signal, the SU remains silent during the data transmission phase since the frame belongs to the PU. Otherwise, the SU starts to exploit the data transmission phase to transmit and receive its own data.

When spectrum sensing phase starts, the SU begins to collect samples, $y[\ell]$, where $\ell=1,2,\dots,N$ with $N$ denoting the maximum number of collected samples during the spectrum sensing phase. If the PU is not occupying this spectrum slot, $y[\ell] = w[\ell]$, where $w[\ell]$ is a zero-mean white Gaussian noise with variance $\sigma^2$. If the PU is occupying the frequency band, the received signal at the SU is $y[\ell]= x[\ell]+w[\ell]$, where $x[\ell] = h s[\ell]$ is the product of the channel coefficient\footnote{Note that $h$ incorporates the multipath components.}, $h$, and the PU signal, $s[\ell]$. The signal $x[\ell]$ is assumed to be an independent and identically distributed (i.i.d.) Gaussian signal with zero mean and variance $P$ \cite{7355347,Ying-Chang08, LULU12,Poor1}. When the PU signal is present, $y[\ell]$  follows  $\mathcal{N}(0,\sigma^2 + P)$, where $\mathcal{N}(\cdot,*)$ denotes a Gaussian distribution with mean $\cdot$ and variance $*$. Otherwise, $y[\ell]$  follows $\mathcal{N}(0,\sigma^2)$.

The spectrum sensing operation has two cases:
\begin{description}
  \item[A)] Detection of the entrance of the PU's signal; or
  \item[B)] Detection of exiting of the PU user, i.e., empty spectrum frame.
\end{description}
In the first case, the two detection hypotheses are defined as
\begin{align}
&H_0^A: y[\ell] = w[\ell],\hspace{0.95 in} \ell=1,\cdots, \tau-1\\
&H_1^A: y[\ell] =  x[\ell]+w[\ell],  \hspace{0.58 in} \ell=\tau,\cdots, N
\end{align}
where $\tau\in[1,N]$, $\tau$ is the time instant at which the PU enters the spectrum and $N$ is the total number of samples collected during the sensing time duration at each frame. When $\tau=1$, this indicates that the PU is already occupying the spectrum when the SU started its sensing time within the frame. Typically, the decision statistic is compared to a threshold to decide on the occupancy status of the spectrum. If the decision raises a flag at sample $N_s$, where $N_s\ge \tau$, the detection delay is $N_d=N_s-\tau$. However if the flag is raised when $N_s<\tau$, this indicates a false-alarm event with $N_f=\mathbb{E}_{f_{0}}[N_s]$ as the mean time to false alarm. In sequential detection, the objective is to minimize $N_d$ and maximize $N_f$. We are interested in evaluating the performance of the CUSUM technique, i.e., calculating the probability of detection and probability of false alarm when a flag is raised. 

In the second case, the two detection hypotheses are defined as
\begin{align}
&H_0^B: y[\ell] = x[\ell]+ w[\ell],\hspace{0.675 in} \ell=1,\cdots, \tau - 1\\
&H_1^B: y[\ell] =  w[\ell],  \hspace{1.05 in} \ell=\tau,\cdots, N.
\end{align}
When $\tau=1$, this indicates that the PU has already exited the spectrum when the SU began its spectrum sensing operation.

\begin{figure}
\centering
\includegraphics[width=3in]{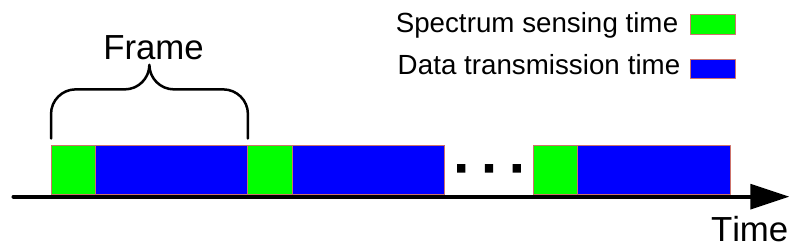}
\caption{Periodic spectrum sensing.}
\label{sys_model}
\end{figure}

\section{QCD based on CUSUM}
The utilization of CUSUM for spectrum sensing under the assumption of full knowledge of  distributions of the signal under the two hypotheses is a common practice in literature \cite{Poor1,7355347}. This assumption provides a performance upperbound benchmark for the case when the variance of the signal under $H_1$ is unknown.
In this letter, we investigate in details the first case, i.e., when the spectrum status changes from vacant to occupied. For the other case, same flow of steps can be used to detect the change of status of the spectrum from occupied to vacant; hence, we only point out the change in formulation to save space and make the discussion concrete. Next, we present necessary background review on CUSUM algorithm before we proceed to present our performance analysis work. 
\subsection{CUSUM Algorithm Background}
CUSUM algorithm is based on LRT \cite{poorbook}. When the spectrum slot is vacant, the collected samples by the SU follow a certain distribution, say distribution $F_0$, with density function $f_0$. Ditto, as the PU starts using the frequency band, the distribution changes to $F_1$ with density $f_1$. In this case, the detection of the entrance of the PU's signal is a sequential change detection problem where the received samples are processed sequentially and the decision statistic is calculated after each sample. The decision on the occupancy status of the spectrum is also made sequentially.
To this end, the log-likelihood ratio is calculated for each sample $y[\ell]$ sequentially through:


\begin{align}
l(y[\ell])&=\ln\left\{\frac{f_1(y[\ell])}{f_0(y[\ell])}\right\}, \notag \\ \label{likelihood}
&=\frac{P y^2[\ell]}{2(P+\sigma^2)\sigma^2}+\frac{1}{2}\ln \left\{\frac{\sigma^2}{P+\sigma^2}\right\},
\end{align}
where $f_\nu(y[\ell])$, $\nu \in \{0,1\}$, is the density function value at sample $y[\ell]$. The Kullback-Leibler divergence of $f_0$ from $f_1$ exhibits a negative drift before the entrance of PU signal and positive drift otherwise \cite{poorbook} and the decision statistic for the CUSUM test can be applied recursively using \cite{Poor1}:
\begin{eqnarray}
\label{eqn13} g_{\ell+1}&=&\max \left\{g_\ell+l(y[\ell+1]),0\right\},\\ 
\nonumber g_0&=&0.
\end{eqnarray}
In the case of changing from occupied to vacant, the same steps are followed with $f_0$ and $f_1$ being swapped with each other.


\subsection{Performance Analysis of CUSUM Algorithm}
Using CUSUM algorithm, after the spectrum status changes, as the number of collected samples increases, the probability of detection increases and, eventually, can reach its maximum value, i.e., 1. However, within the paradigm of periodic sensing, depicted in Fig. \ref{sys_model}, the total number of collected samples is bounded by the periodic sensing time. To this end, when applying CUSUM algorithm in cognitive radio applications where the size of the detection window is fixed, i.e., finite number of collected samples is used, the receiver operating characteristics (ROC), determined by the probability of false alarm and the probability of detection, are key performance metrics related to deciding on the status of the spectrum.



In \textit{Proposition 1}, we derive closed-form expressions for the detection and false-alarm probabilities of the decision statistic of the CUSUM test, $g_{\ell+1}$. The distribution of the received signal under the two hypotheses is the same as given in Section~\ref{sec_sys_model}.
\begin{prop}
The false-alarm probability for the $(\ell+1)^{\rm th}$ sample, $P_{f_{\ell+1}}$, is given by
\begin{align}
P_{f_{\ell+1}}
&=\left(1-\sum_{r=1}^{\ell+1}{\frac{\gamma\left(\frac{\ell+2-r}{2},\frac{\zeta}{2\sigma^2} \right)}{\Gamma\left(\frac{\ell+2-r}{2}\right)}}\right) \nonumber \\ &\left(\prod_{j=1}^{\ell}{\sum_{r=1}^{j}\frac{\gamma\left(\frac{j-r+1}{2},\frac{\zeta}{2\sigma^2} \right)}{\Gamma\left(\frac{j-r+1}{2}\right)}}\right). \label{P_f_F}
\end{align}
The total false-alarm probability is given by
\begin{align}
P_f=\sum_{\ell=1}^{\tau - 1}{P_{f_{\ell+1}}} \label{P_f_tot}
\end{align}
In addition, the detection probability is given by
\begin{align}
P_{d_{\ell+1}}&=\left(1-\sum_{r=\tau}^{\ell+1}{\frac{\gamma\left(\frac{\ell+2-r}{2},\frac{\zeta}{2(P+\sigma^2)} \right)}{\Gamma\left(\frac{\ell+2-r}{2}\right)}}\right)\nonumber \\ &\left(\prod_{j=\tau}^{\ell}{\sum_{r=\tau}^{j}\frac{\gamma\left(\frac{j-r+1}{2},\frac{\zeta}{2(P+\sigma^2)} \right)}{\Gamma\left(\frac{j-r+1}{2}\right)}}\right),\label{P_d_F}
\end{align}
\noindent where 
$\zeta=\frac{\lambda-(\ell+2-r)c_2}{c_1}$, $\lambda$ is the threshold,
$c_1 = \frac{P}{2(P+\sigma^2)\sigma^2}$, $c_2 = \frac{1}{2} \ln {\frac{\sigma^2}{P+\sigma^2}}$, $\gamma(\cdot,\cdot)$ is the lower incomplete gamma function, and $\Gamma(\cdot)$ is the gamma function. The total detection probability is given by
\begin{align}
P_d=\sum_{\ell=\tau}^{N}{P_{d_{\ell+1}}}. \label{P_d_tot}
\end{align}
\end{prop}
\begin{proof}
 The proof of \textit{Proposition 1} is provided in Appendix~\ref{app:proof1}.
\end{proof}
Using our closed-form expressions and for a desired $P_f$ or $P_d$ under finite sensing time, the decision statistic is compared to the threshold ($\lambda$), which can be calculated according to the desired performance metric.
\section{Simulation Results}
\begin{figure}
	\centering
	\includegraphics[width=1 \columnwidth]{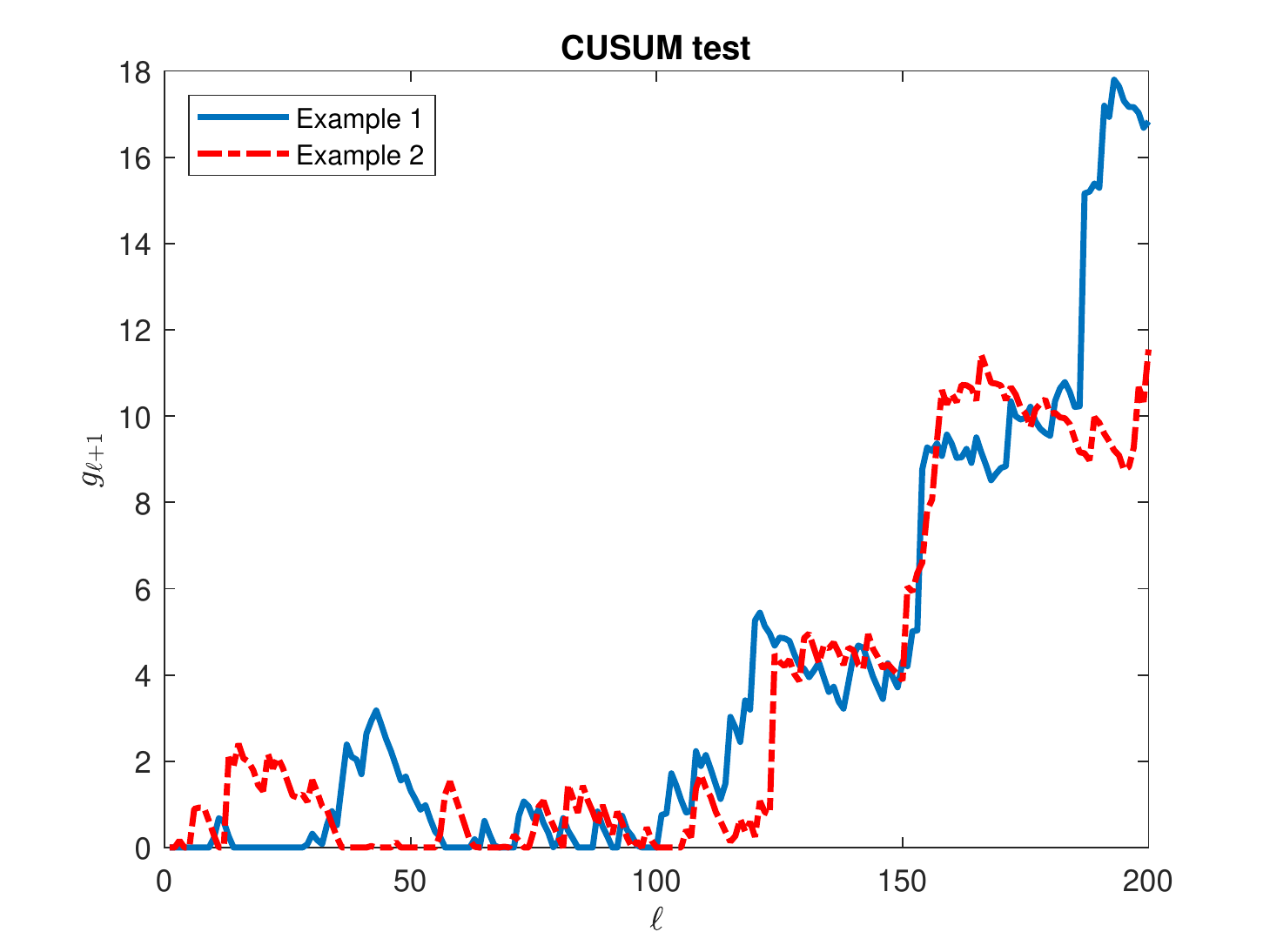}
	\caption{Examples of CUSUM test for 200 samples with PU entering the spectrum at the $100^{\rm th}$ sample, i.e., $\tau = 100$.}
	\label{CUSUM}
\end{figure}
We provide simulation results for the detection of the entrance of the PU signal. We compare the numerical calculations for $P_f$ and $P_d$ with our derived analytical approximations presented in (\ref{P_f_F}) to (\ref{P_d_tot}). We plot total probability of detection at various samples after the entrance of the PU signal, i.e., $P_{d}=\sum_{i=\tau}^{L}{P_{d_{i}}}$ vs $P_f$ for different signal to noise ratio (SNR) values, where $L$ is the test sample index which is in the range  $\{\tau,\tau+1,\cdots, N\}$. We run simulations for 200 samples with the first 100 belonging to $H_0$ and follow $F_0$ distribution, i.e., $\tau=100$, and the second 100 samples belong to the PU and follow $F_1$. Figure \ref{CUSUM} depicts two examples of the calculated decision statistic, $g_{\ell+1}$, for the 200 samples at SNR = 0 dB. Remember that both distributions are zero mean Gaussian with different variances as stated earlier. It is shown that as the PU enters the spectrum, the decision statistic, $g_{\ell+1}$, which processes the samples sequentially, starts to increase. Hence, $g_{\ell+1}$ is compared to a preset threshold and a decision about the status of the spectrum can be made after each received sample.

We run Monte Carlo simulations to numerically evaluate $P_f$ and $P_d$. Figures \ref{results1}, \ref{results3} and \ref{results2} show the ROC curves for both numerical and analytical results for $L=\tau + 20$, $L=\tau + 40$ and $L=\tau+60$ at SNR = 3, 0 and -3 dB, respectively. As shown in the figures, the derived analytical expressions provide close results to the numerical evaluations. 
In particular, the difference between the analytical approximation and numerical results is less than $5\%$ across the entire range of $P_f$ and $P_d$ for the presented operational SNR levels.
\begin{figure}
	\centering
	\includegraphics[width=1 \columnwidth]{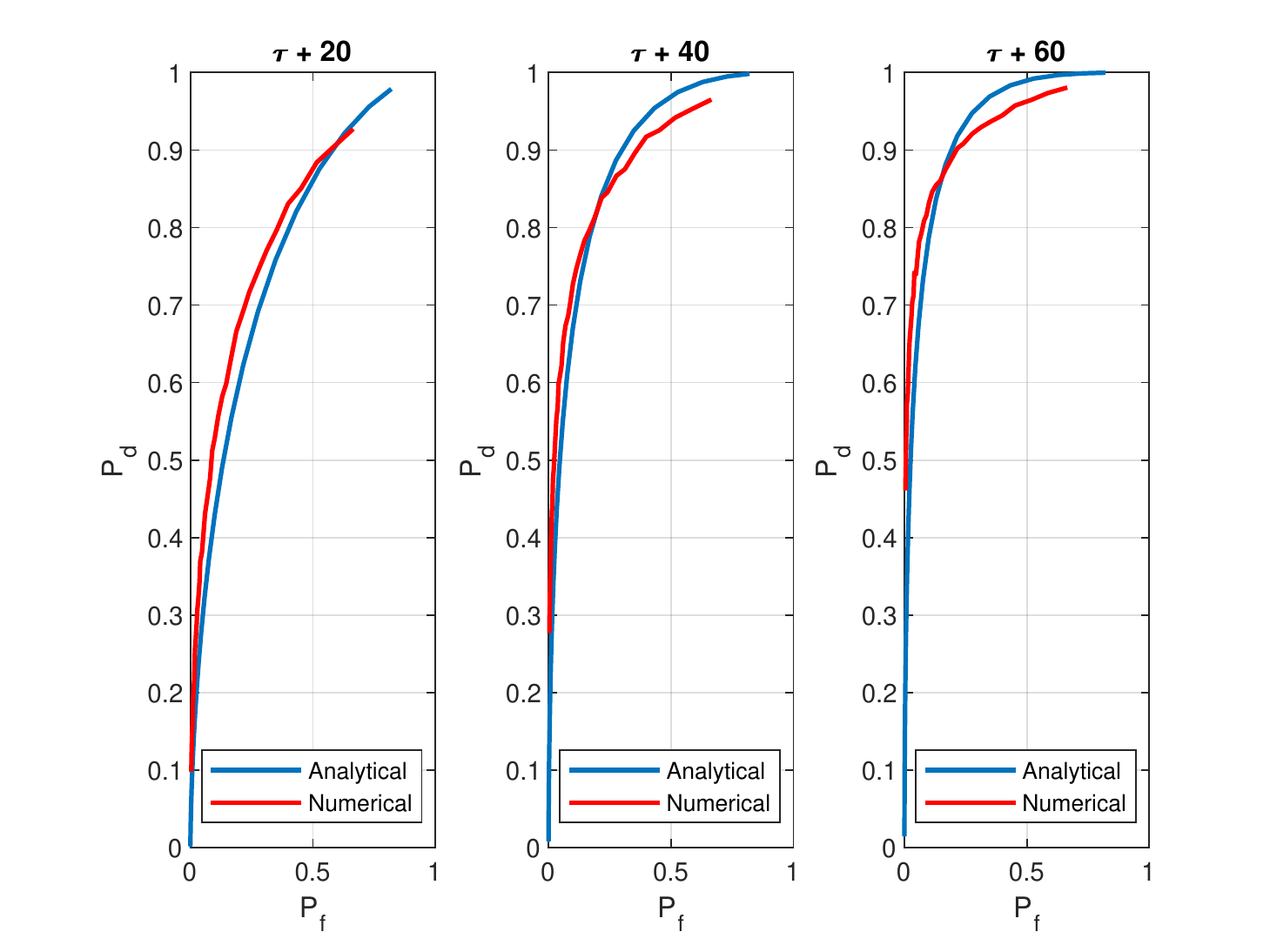}
	\caption{ROC curves for numerical and analytical simulations for SNR = -3 dB.}
	\label{results1}
\end{figure}
\begin{figure}
	\centering
	\includegraphics[width=1 \columnwidth]{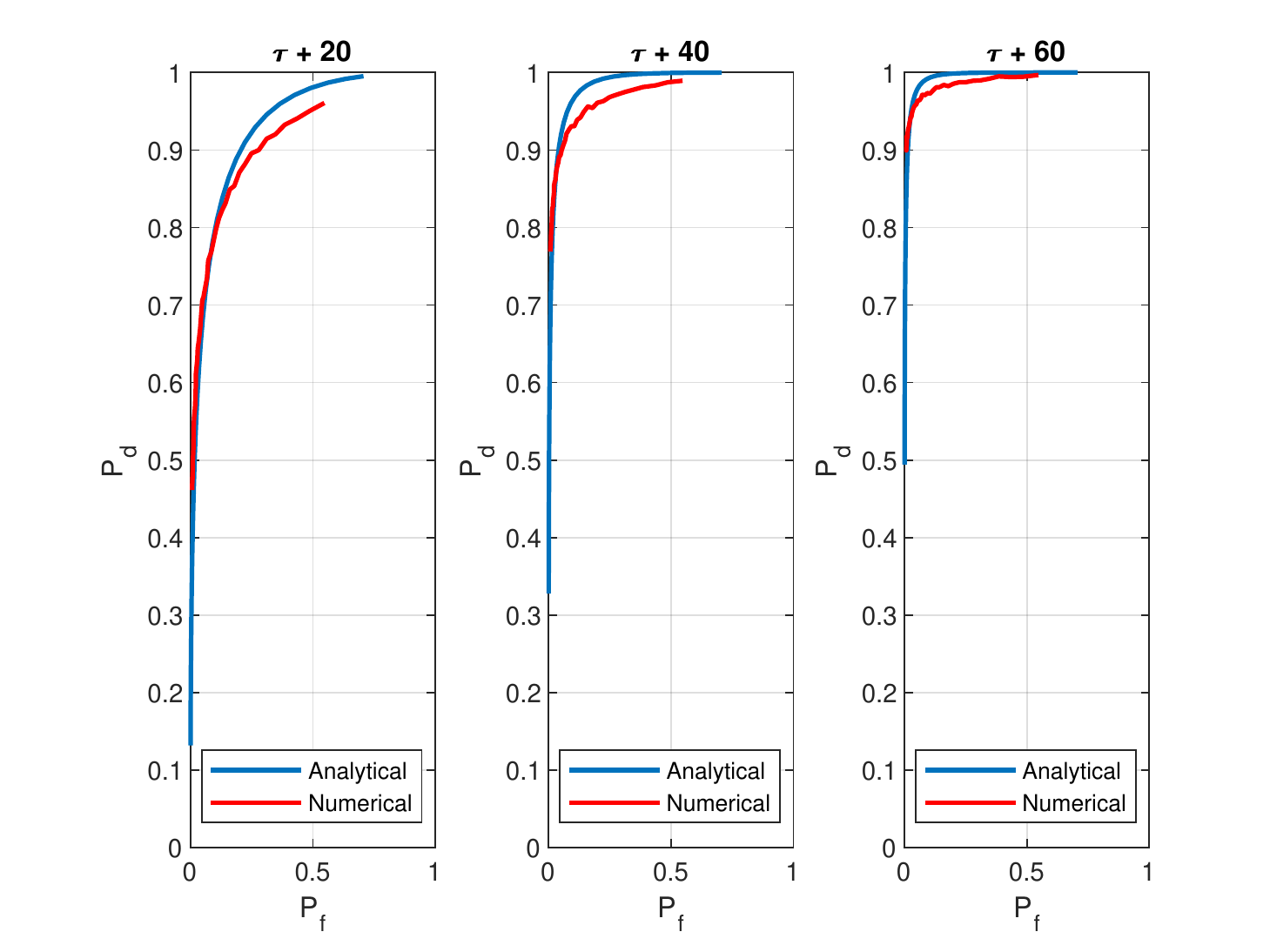}
	\caption{ROC curves for numerical and analytical simulations for SNR = 0 dB.}
	\label{results3}
\end{figure}
Our approximation, which is based on the independence assumptions between the correlated random variables, $Z_{\ell}$ (presented in the appendix), yields results very close to the exact ones. As expected, as the SNR level increases, $P_d$ increases. This is due to better separation between the two distributions when calculating the likelihood ratio term in (\ref{likelihood}). Similarly, as $L$ increases, $P_d$ increases. This is because as $L$ increases, more accumulation of the recursive decision statistic, as presented in (\ref{eqn13}) occurs after the entrance of the PU's signal. Hence, when the decision statistic exceeds the threshold, it is more likely that it is due to a detection of the PU's entrance.

The average time between false alarm and longest detection delay parameters discussed in, for example, \cite{Poor1}, assume infinite number of samples can be collected and processed sequentially. Hence, assume that $P_f$ will approach $0$ and $P_d$ will approach $1$ eventually. This is fundamentally different from our current application of assuming finite sensing duration. To this end, our closed-form expressions for $P_f$ and $P_d$ are very essential in estimating the delay needed to achieve a specific $P_d$ within the finite sensing duration and whether or not the required ROC of combined $P_d$ and $P_f$, which are typically $0.9$ and $0.1$, respectively, can be achieved under the current sensing parameters.
\begin{figure}
\centering
\includegraphics[width=1 \columnwidth]{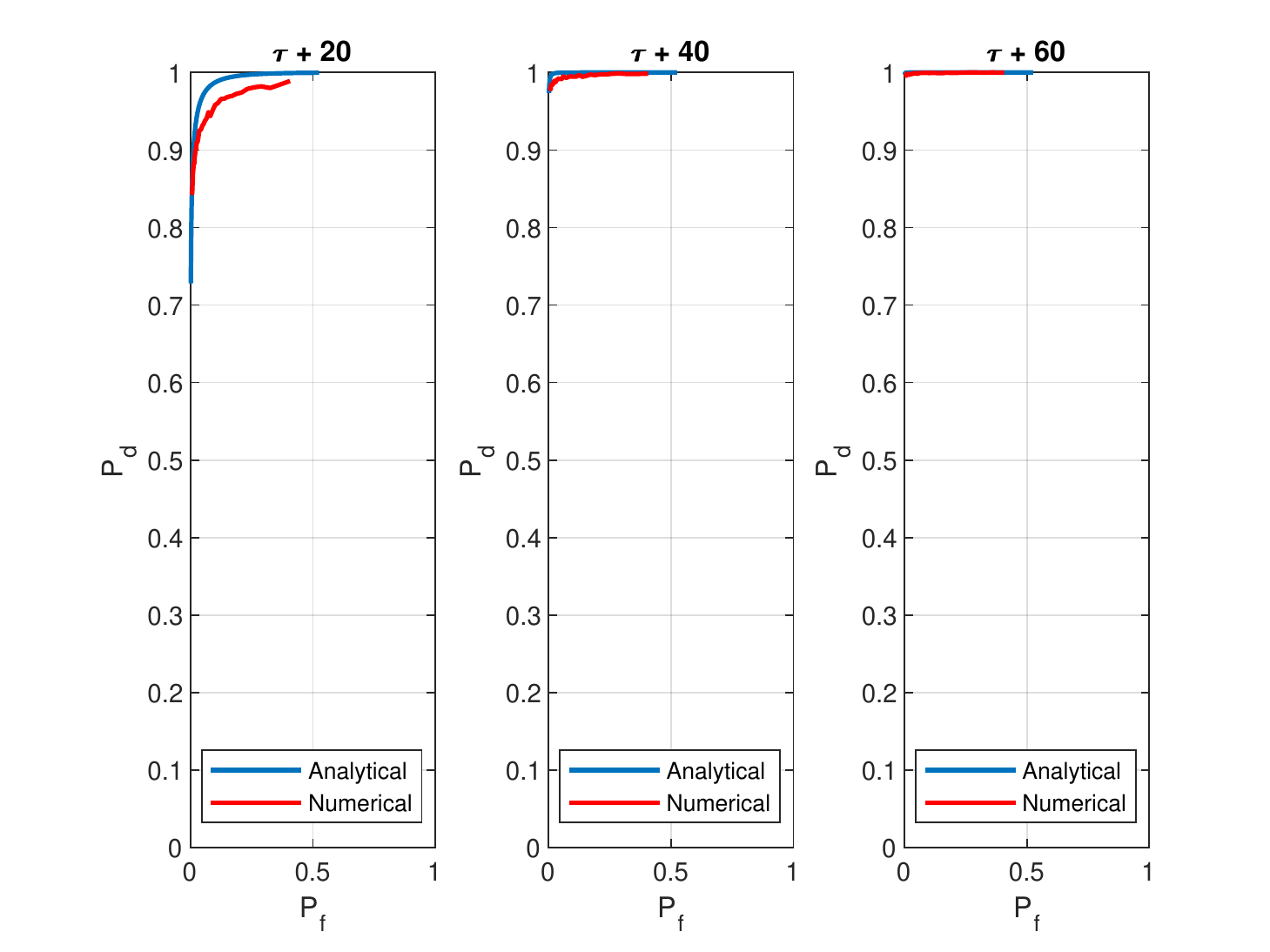}
\caption{ROC curves for numerical and analytical simulations for SNR = 3 dB.}
\label{results2}
\end{figure}
\section{Conclusions}
In this paper we derived closed-form expressions for the probability of false alarm and probability of detection for QCD-CUSUM sequential test. Spectrum sensing is used as an application of CUSUM test with detecting the entrance of PU signal as the example. Through simulation comparison between analytical and numerical results, we showed that our derived expressions provide a very close approximation. The provided results can be used in all applications that require channel sensing and will simplify the optimization formulation procedures to achieve better performance for various applications.


%





\ifCLASSOPTIONcaptionsoff
  \newpage
\fi



%
\appendices 
\section{Proof for Proposition 1}\label{app:proof1}
To derive closed-form expression for the probability of false alarm and the probability of detection, we need to define the probability distribution of the decision statistic defined in (\ref{eqn13}). Hence, Our starting point in the derivation is $ g_{\ell+1}$ in (\ref{eqn13}), which is defined using a $\max$ operation. Therefore, let us first consider two random variables (RVs) $V_1$ and $V_2$, the probability distribution function of their maximum, $U=\max[V_1,V_2]$, $F_U(u)$, can be given by
\begin{align}
  F_U(u)=\Pr\left\{V_1 \leq u, V_1>V_2\right\} +\Pr\left\{V_2 \leq u, V_1\leq V_2\right\}. \label{eqn147}
\end{align}
To this end and noting that $g_{\ell+1}$ in (\ref{eqn13}) is the maximum of a RV and zero, and defining this RV as $Z_{\ell+1}=g_\ell+l(y[\ell+1])$,  therefore by substituting in (\ref{eqn147}) for the probability distribution of $g_{\ell+1}$ at a threshold $\lambda$, we have
\begin{align}
  F_{g_{\ell+1}}(\lambda)&=\Pr\left\{Z_{\ell+1} \leq \lambda, Z_{\ell+1}>0\right\} \nonumber \\&+\Pr\left\{0 \leq \lambda, Z_{\ell+1}\leq 0\right\} \label{eq_max_two}
\end{align}
Here, we always have the threshold $\lambda>0$, hence (\ref{eq_max_two}) is now
\begin{align}
  F_{g_{\ell+1}}(\lambda)&=\Pr\left\{0<Z_{\ell+1} \leq \lambda\right\}+\Pr\left\{Z_{\ell+1}\leq 0\right\}.\label{eq_max_two2}
\end{align}
Note that the probability that a RV $U$ lies in an interval $[u_1, u_2]$, where $u_1<u_2$, can be given by
\begin{align}
    \Pr(u_1<U\leq u_2) = F_U(u_2) - F_U(u_1). \label{eqn478}
\end{align}
\noindent Using (\ref{eqn478}) to substitute for the first term in (\ref{eq_max_two2}) and noting that $F_{Z_{\ell+1}}(0) =\Pr\left\{Z_{\ell+1}\leq 0\right\}$,  (\ref{eq_max_two2}) now becomes
\begin{align}
  F_{g_{\ell+1}}(\lambda)&= F_{Z_{\ell+1}}(\lambda)- F_{Z_{\ell+1}}(0) + F_{Z_{\ell+1}}(0) \notag\\
  &=F_{Z_{\ell+1}}(\lambda)
\end{align}

\noindent where $F_{Z_{\ell+1}}$ represents the probability distribution of $Z_{\ell+1}$. The quantity $Z_{\ell+1}$ represents the likelihood ratios summation up to the sample $\ell+1$ with the possibility that each $g$ will be reset to zero at any sample inside the $\ell+1$ ($g_\ell=\max[g_{\ell-1}+l(y[\ell]),0]$) samples. It is noteworthy that, combinations or positions of each zero incident do not have an impact on the results, only position of the last occurring zero matters. Hence, $Z_{\ell+1}$ has $\ell+1$ possibilities. For example, if the output of the maximization with zero process resulted in no zero, $Z_{\ell+1}=\sum_{j=1}^{\ell+1}{l(y[j])}$. If a zero occurred at the first sample, $Z_{\ell+1}=\sum_{j=2}^{\ell+1}{l(y[j])}$ and so on. We have
\begin{align}
F_{Z_{\ell+1}}&=\Pr\left\{Z_{\ell+1}\leq \lambda\right\} \notag\\
&=\Pr\left\{\sum_{j=1}^{\ell+1}{l(y[j])}\leq \lambda\right\}+\Pr\left\{\sum_{j=2}^{\ell+1}{l(y[j])}\leq \lambda\right\} \nonumber \\
&+\cdots+\Pr\left\{{l(y[\ell+1])}\leq \lambda\right\} \notag\\
&=\sum_{r=1}^{\ell+1}{\Pr\left\{\sum_{j=r}^{\ell+1}{l(y[j])}\leq \lambda\right\}}.
\end{align}
\noindent Note that
\begin{align}
\sum_{j=r}^{\ell+1}{l(y[j])}\leq \lambda&=c_1\left(\sum_{j=r}^{\ell+1}{y^2[j]}\right)+(\ell+2-r)c_2\leq \lambda \notag\\
&=\sum_{j=r}^{\ell+1}{y^2[j]}\leq \zeta. \label{eqn_sum_chi}
\end{align}
\noindent where $\zeta=\frac{\lambda-(\ell+2-r)c2}{c_1}$, $c_1 = \frac{P}{2(P+\sigma^2)\sigma^2}$ and $c_2 = \frac{1}{2} \ln {\frac{\sigma^2}{P+\sigma^2}}$. Recalling that the collected samples are Gaussian RVs, the samples $y^2[\ell]$ are Chi-square RVs. Hence, $\sum_{j=r}^{\ell+1}{y^2[j]}$ is a Chi-square RV with $\ell+2-r$ degrees of freedom since it is a summation of Chi-square RVs. Consequently,
\begin{align}
  \Pr\left\{\sum_{j=r}^{\ell+1}{l(y[j])}\leq \lambda\right\}&=\Pr\left\{\sum_{j=r}^{\ell+1}{y^2[j]}\leq \zeta\right\} \notag\\
  &=\frac{\gamma\left(\frac{\ell+2-r}{2},\frac{\zeta}{2\sigma^2} \right)}{\Gamma\left(\frac{\ell+2-r}{2}\right)},
\end{align}
\noindent where $\gamma(\cdot,\cdot)$ denotes the lower incomplete gamma function and $\Gamma(\cdot)$ is the gamma function. This leads to
\begin{align}
F_{Z_{\ell+1}}=\sum_{r=1}^{\ell+1}{\frac{\gamma\left(\frac{\ell+2-r}{2},\frac{\zeta}{2\sigma^2} \right)}{\Gamma\left(\frac{\ell+2-r}{2}\right)}}. \label{eq_fxi}
\end{align}
The probability of false alarm for the $(\ell+1)^{\rm th}$ sample where $\ell\in[1:\tau - 1]$, is given by
\begin{align}
P_{f_{\ell+1}}&=\Pr\left\{Z_{\ell+1}>\lambda, \max\left[Z_{1},\cdots,Z_{\ell}\right]<\lambda\mid H_0\right\}. \label{eqn1145}
\end{align}
Note that the joint distribution of two independent random variables $V_1$ and $V_2$ can be given by
\begin{align}
    F_{V_1,V_2} (v_1,v_2)&= \Pr\left\{V_1\leq v_1, V_2\leq v_2\right\} \notag \\
    &= F_{V_1}(v_1) \times F_{V_2}(v_2). \label{eqn746}
\end{align}
Moreover, the distribution of the RV $V = \max\left[V_1, V_2, \ldots, V_n \right]$, where $V_1$ and $V_2$ to $V_n$ are independent RVs, is given by
\begin{align}
    F_V(v) &= \Pr\left\{\max\left[V_1, V_2, \ldots, V_n \right] \leq v \right\} \notag \\
         &= \Pr\left\{V_1 \leq v \right\} \times \Pr\left\{V_2 \leq v \right\} \times \ldots \times  \Pr\left\{V_n \leq v \right\} \notag \\
         &= \prod_{i=1}^{n} {F_{V_i} (v)}. \label{eqn598}
\end{align}
Using (\ref{eqn746}) and (\ref{eqn598}) and substituting in (\ref{eqn1145}), we obtain the expressions in (\ref{P_f_F}) and (\ref{P_f_tot}).

\noindent The probability of detection for the $(\ell+1)^{\rm th}$ sample where $\ell\in[\tau:N - 1]$, is given by
 \begin{align}
P_{d_{\ell+1}}&=\Pr\left\{Z_{\ell+1}>\lambda, \max\left[Z_{\tau},\cdots,Z_{\ell}\right]<\lambda\mid H_1\right\}. \label{eqn5547}
\end{align}
Using (\ref{eqn746}) and (\ref{eqn598}) and substituting in (\ref{eqn5547}), we obtain the closed-form expressions in (\ref{P_d_F}) and (\ref{P_d_tot}). 

It must be noted that $Z_\ell$ are correlated RVs, nevertheless, we assume independence as an approximation to make the closed-form expressions derivations feasible. Due to the recursive nature of $Z_\ell$'s, deriving a closed-form expression with dependence assumption might be infeasible, as also stated in \cite{7355347}. Nevertheless, we will show that the independence approximation provides very close results. 



\bibliographystyle{IEEEtran}
\bibliography{references_QD_Journal}
\end{document}